\documentclass[sigconf]{aamas} 
\usepackage{balance} % for balancing columns on the final page
\usepackage{tikz}

\setcopyright{ifaamas}
\acmConference[AAMAS '22]{Proc.\@ of the 21st International Conference
on Autonomous Agents and Multiagent Systems (AAMAS 2022)}{May 9--13, 2022}
{Auckland, New Zealand}{P.~Faliszewski, V.~Mascardi, C.~Pelachaud,
M.E.~Taylor (eds.)}
\copyrightyear{2022}
\acmYear{2022}
\acmDOI{}
\acmPrice{}
\acmISBN{}

\acmSubmissionID{143}

\title[Incentives to Invite Others to Form Larger Coalitions]{Incentives to Invite Others to Form Larger Coalitions}

\author{Yao Zhang}
\affiliation{
  \institution{ShanghaiTech University}
  \city{Shanghai}
  \country{China}}
\email{zhangyao1@shanghaitech.edu.cn}

\author{Dengji Zhao}
\affiliation{
  \institution{ShanghaiTech University}
  \city{Shanghai}
  \country{China}}
\email{zhaodj@shanghaitech.edu.cn}

\begin{abstract}
We study a cooperative game setting where the grand coalition may change since the initial players can invite more players. We focus on monotone games, i.e., adding more players to the grand coalition is not harmful. We model the invitation relationship as a directed acyclic graph. Our goal is to design a reward distribution mechanism for this new cooperative game setting such that players are incentivized to invite new players. In this paper, we propose the weighted permission Shapley value (inspired by permission structure and the weighted Shapley value) to achieve the goal. Our solution offers the very first attempt to incentivize players to invite more players to form a larger coalition via their private invitations in cooperative settings.
\end{abstract}


\newcommand{\BibTeX}{\rm B\kern-.05em{\sc i\kern-.025em b}\kern-.08em\TeX}

\begin{document}

%%% The following commands remove the headers in your paper. For final 
%%% papers, these will be inserted during the pagination process.

\pagestyle{fancy}
\fancyhead{}

%%% The next command prints the information defined in the preamble.

\maketitle 

%%%%%%%%%%%%%%%%%%%%%%%%%%%%%%%%%%%%%%%%%%%%%%%%%%%%%%%%%%%%%%%%%%%%%%%%

\section{Introduction}
Cooperative game is a classical research topic studied in game theory~\cite{nisan2007algorithmic,narahari2014game}. The study focused on games where a fixed set of players forms coalitions to share rewards. A common assumption behind the game is that more players collaborating create a higher reward and the goal is to incentivize all players to collaborate. Then one immediate question is how we can gather all the players initially, which is not tackled in traditional settings. With the rapid development of Internet and social media, people now can remotely form groups to collaborate on projects via, for example, crowdsourcing platforms (e.g. Amazon Mechanical Turk) and Q\&A platforms (e.g. Stack Overflow). However, the players in these platforms are relatively independent and the tasks assigned to them are also independent. Thus the incentives for the players to actively participate in these platforms are not strong.%However, we may still fail to get enough participants via these online platforms, because participants need to constantly check what projects are available by themselves. 

In this paper, we want to design a mechanism where existing participants are motivated to invite others to join a project via their private social connections. This kind of mechanism is highly demanded in practice. For instance, when a research team wants to collect a large scale dataset, they may propagate the tasks via social networks to seek more data providers~\cite{singla2015information,jeong2013crowd}. When a company wants to select a set of users for a product trial, it may diffuse the recruitment via social networks to find the target users. Again, we may use social networks to find a missing person quickly. Among them, the most famous red balloon challenge hosted by DARPA is a well-known example where collaboration formed via social networks played an essential role~\cite{pickard2011time}.

Similar to the classical cooperative games, the key challenge here is also to design a reward distribution mechanism for all players in the coalition. The difference is that in our setting players are connected and a player cannot join a coalition without the invitation from her neighbours (simply because the player is not aware of the collaboration without the others' invitation). \emph{More precisely, our challenge is to design a reward distribution mechanism such that the current players of the coalition are incentivized to invite their neighbours to join the coalition.}

The Shapley value is a classical solution concept in cooperative games to distribute rewards~\cite{shapley1953value}. However, in the calculation of the Shapley value, all the players are symmetric, which is not the case in our setting because some players are invited by the others and they cannot be treated equally. Therefore, it is easy to show that the Shapley value cannot be directly applied here to incentivize players to invite others.

To combat this problem, we combine the concepts of permission structure and weighted Shapley value for the first time to build our solution. The weighted Shapley value is the very first concept that applies asymmetry to cooperative games~\cite{chun1991symmetric}. Kalai and Samet~\cite{kalai1987weighted} offered the idea of the weight system and an axiomatic characterization of the weighted Shapley values. A recent understanding of the weighted Shapley value is illustrated in~\cite{radzik2012new}, which gave new axioms and provided simpler formula for weighted Shapley value. However, the asymmetry introduced by the weighted Shapley value cannot reflect the structure of the invitations among the players in our setting.

The concept of permission structure seems a closer solution to our problem, which was first introduced in~\cite{gilles1992games} and~\cite{gilles1999cooperative}. In the permission structure, players need permissions from other players before they are allowed to cooperate, where the permissions are very similar to the players' invitations in our model. In our model, a player is not aware of the game without someone's invitation. Gilles et al.~\cite{gilles1992games,gilles1999cooperative} considered the cases where each player has to get permissions from all or at least one of her superiors, which is called the conjunctive and disjunctive approach respectively. The axiomatic characterization of the Shapley value under these two approaches is given in~\cite{van1996axiomatizations} and~\cite{van1997axiomatization}. %The Banzhaf permission value of the two approaches was described in~\cite{van2010axiomatizations}, which defined the importance or power index for each player in a network based on the permission structure. %In this paper, we extend the permission structure by mixing the conjunctive and disjunctive approaches, which can deal with instances like: player $a$ has to get permissions from superiors $b$ and $c$, or just superior $d$.

Against this background, our solution is a novel combination of the weighted Shapley value and the permission structure. We use permission structure to represent the priorities between an inviter and invitee. We further assign different weights to the players to control the importance of their priorities. The well-known winning solution for the DARPA red balloon challenge is a special case of our solution~\cite{pickard2011time}. We expect that our solution will have very promising applications for task allocations via social networks such as crowdsourcing and question answering. Our contributions advance the state of the art in the following ways:
\begin{itemize}
\item We formally model the (social) connections between players in a cooperative game and, for the first time, define the concept of \emph{diffusion incentive compatibility (DIC)} for players to utilize their connections to gather more players.
\item We define a weighted permission Shapley value as a reward distribution mechanism to achieve DIC.
\item We also formally model the query network as a diffusion cooperative game and show the only solution to it is the weighted permission Shapley value, which for the first time explains in theory why the winning solution of DARPA red balloon challenge worked and extends their solution to DAGs.
\end{itemize}

There are some interesting related works investigating information diffusion on social networks in non-cooperative settings~\cite{emek2011mechanisms,li2017mechanism,zhao2018selling,sinha2019incentivizing}. Emek et al.~\cite{emek2011mechanisms} firstly investigated the setting of multi-level marketing, where players promote a sale of products to their neighbours in a tree. For incentivizing players to invite friends to an auction, Li et al.~\cite{li2017mechanism} and Zhao et al.~\cite{zhao2018selling} proposed the very first mechanisms. Our model shares a similar motivation, but cannot be handled with their techniques as they focused on the non-cooperative setting and only a few players can get reward.

The remainder of the paper is organized as follows. Section 2 gives a formal description of the model. Sections 3 establishes the family of weighted permission Shapley value to incentivize diffusion in forests and Section 4 demonstrates its applicability in query networks. Section 5 extends the result to general DAGs.
%We conclude our work in Section 6.

\section{The Model}
We study a cooperative game where players are connected to form a network and not all players are aware of the game initially. In real-world applications, their connections can represent friendship or leadership. A person who is in the coalition can invite her friends who are not in the coalition yet to join. Without the invitation, these friends will not be informed of the coalition. We investigate the reward distribution mechanism in this setting to incentivize existing players to invite new players to join the coalition.

Formally, let $N = \{1,2,\dots,n\}$ be the set of all connected players in the underlying network. We model the network as a directed acyclic graph (DAG) $G = (N, E)$. Each edge $e = (x,y) \in E$ indicates that player $x$ can invite $y$. There is a special player set $\mathcal{I}\subseteq N$ who are in the coalition initially without invitation. We call $\mathcal{I}$ the initial set and the invitation has to start from the initial players. For each player $i\in N$, let $p_i = \{ j \mid (j,i)\in E \}$. We have $p_i = \emptyset$ if and only if $i \in \mathcal{I}$. We may assume w.l.o.g. that all players can be reached from at least one of the initial players in the underlying network. Let $\theta_i = \{ j \mid (i,j)\in E \}$ be private type of player $i$, which is the set of players who can be invited by $i$. Since the mechanism does not know the real underlying network, the invitation process can be modelled as each player $i$ reporting her own type. Let $\theta_i'\subseteq \theta_i$ is the report type of player $i$, i.e., the actual player set invited by $i$. Given any report profile $\theta' = (\theta_1',\dots, \theta_n')$, there is a directed graph induced by $\theta'$ denoted by $G(\theta') = (N, E(\theta'))$, where $E(\theta')= \{(i,j) \mid i\in N, j\in \theta_i'\}$. Let $J_{\mathcal{I}}(G(\theta'))$ be the set of players who can be reached from at least one player in $\mathcal{I}$ in $G(\theta')$. It is clear that only the players in $J_{\mathcal{I}}(G(\theta'))$ can actually join the game, because the others cannot receive the proper invitation started from the initial players (in practice, this means that the others will not be informed about the game at all).
% Formally, let $N = \{1,2,\dots,n\}$ be the set of all connected players in the underlying network. There is a special player set $\mathcal{I}\subseteq N$ who are in the coalition initially without invitation. We call $\mathcal{I}$ the initial set. Since the invitation has to start from the initial players, we model the network as a directed acyclic graph (DAG) $G = (N, \mathcal{E})$. Each edge $e = (x,y) \in \mathcal{E}$ indicates that player $x$ can invite $y$. For each player $i\in N$, let $p(i) = \{ j \mid (j,i)\in \mathcal{E} \}$. We have $p(i) = \emptyset$ if and only if $i \in \mathcal{I}$. We may assume w.l.o.g. that all players can be reached from at least one of the initial players.

There is a non-negative and monotone characteristic function $v:2^N\rightarrow \mathbb{R}$, s.t., $v(\emptyset) = 0$ and $v(S) \leq v(T)$ for all $S\subseteq T \subseteq N$. The monotone property is necessary in our setting; otherwise, there is no need to invite more people to join the coalition if fewer people can do better. Note that the definition of $v$ does not consider the connections between players, simply because the connections are their private information, which is what we want to discover. Let $\Theta(N)$ and $\mathcal{V}(N)$ be the space of all type profiles and all characteristic functions satisfying our setting respectively. We define the reward distribution mechanism as follows.

\begin{definition}
A \textbf{reward distribution mechanism} $\mathcal{M}$ is defined by a reward policy $\Phi = \{\phi_i\}_{i\in N}$, where each $\phi_i:\Theta(N) \times \mathcal{V}(N) \rightarrow \mathbb{R}$ assigns the reward to player $i \in N$. Moreover, for all $\theta'\in \Theta(N)$ and all $v\in \mathcal{V}(N)$, $\phi_i(\theta', v) = 0$ if $i\notin J_{\mathcal{I}}(G(\theta'))$, where $\mathcal{I}$ is the initial set.
% A \textbf{reward distribution mechanism} $\mathcal{M}$ is defined by $\mathcal{M}:\mathcal{G}(N) \times \mathcal{V}(N) \rightarrow \mathbb{R}^{|N|}$, where $\mathcal{M}(G, v) = (\phi_1,\dots,\phi_{n})$ assigns the reward $\phi_i$ to each player $i \in N$, given $G\in \mathcal{G}(N)$ and $v\in \mathcal{V}(N)$.
\end{definition}

A desirable property of the reward distribution mechanism is to distribute exactly what the coalition of all participated players can generate. This is called \emph{efficiency}.

\begin{definition}
A reward distribution mechanism $\mathcal{M}$ is \textbf{efficient} if for all $\theta' \in \Theta(N)$ and all $v\in \mathcal{V}(N)$, we have %$\sum_{i\in N} \mathcal{M}(G, v)_i = v(N)$.
\[ \sum_{i\in N} \phi_i(\theta', v) = v\left( J_{\mathcal{I}}(G(\theta')) \right) \]
\end{definition}

Other than efficiency, we also want to incentivize all players who are already in the coalition to invite all their neighbours to join the coalition. This is the key property we want to achieve here, which requires that inviting all neighbours is a dominant strategy for all players. We call this property \emph{diffusion incentive compatibility}. Let $\Theta_i$ be the type space of player $i$, $\theta_{-i}$ be the report profile of players other than $i$, and $\Theta_{-i}$ be the type space of players other than $i$. We define diffusion incentive compatibility as follows.

\begin{definition}
A reward distribution mechanism $\mathcal{M}$ is \textbf{diffusion incentive compatible (DIC)} if for all $i\in N$ with type $\theta_i\in \Theta_i$, all $\theta_i'\subseteq \theta_i$, all $\theta_{-i}'\in \Theta_{-i}$ and all $v\in \mathcal{V}(N)$, we have
\[ \phi_i((\theta_i, \theta_{-i}'), v) \geq \phi_i((\theta_i', \theta_{-i}'), v) \]
% A reward distribution mechanism $\mathcal{M}$ is \textbf{diffusion incentive compatible (DIC)} if for all $i\in N$ and for all edge set $Y\subseteq \{(i,y)\mid (i,y)\in \mathcal{E}\}$,
% \[ \mathcal{M}(G, v)_i \geq \mathcal{M}(G',v)_i  \]
% where $G' = (N',\mathcal{E}')$ is the graph derived from $G$ by removing the players who are not reachable from any of the initial players without Y. %after removing $Y$ and players cannot be informed of without $Y$, i.e., $N' = \{ i\mid$ if there exists a sequence $j_1, j_2, \dots, j_m$, where $j_1 \in \mathcal{I}, j_m = i$ and $(j_k, j_{k+1}) \in \mathcal{E}\setminus Y$ for all $1\leq k< m$ $\}$ and $\mathcal{E}' = \{(i,j)\mid i,j\in N' \text{ and } (i,j)\in \mathcal{E} \}$.
\end{definition}

Finally, we consider a property of structural fairness that guarantees a player can gain reward at least as much as a fixed proportion of the reward gained by her invitees. Since the game is monotone, it is reasonable to give a promise of fairness compared to neighbours to all players before their invitations.

\begin{definition}
A reward distribution mechanism $\mathcal{M}$ has the property of \textbf{$\gamma$-structural fairness ($\gamma$-SF)} if for all $v\in \mathcal{V}(N)$, all $\theta'\in \Theta(N)$ and all $i,j\in J_{\mathcal{I}}(G(\theta'))$ with $j\in \theta_i'$, we have $\phi_i(\theta', v) \geq \gamma \phi_j(\theta', v)$.
\end{definition}

\section{Diffusion Incentives in a Forest}
\label{section:tree}
In this section, we first investigate the solution to satisfy efficiency and diffusion incentive compatibility when the network $G$ is a forest. An instance of the cooperative game in forests is illustrated in Example~\ref{example:1} below.

\begin{example}
\label{example:1}
Consider the case illustrated in Figure~\ref{figure:1}, where $N=\{1,2,3,4\}$. Suppose the initial coalition is $\mathcal{I}=\{1,2\}$. To form a larger coalition, agents 1 and 2 can invite their friends 3 and 4. Suppose the characteristic function $v$ is defined as
\[
v(S) = \mathbb{I}(\{1, 3\}\cap S \neq \emptyset) + \mathbb{I}(4\in S)
\]
for all $S\subseteq N$, where $\mathbb{I}(\cdot)$ is the indicator function.
%To form a larger coalition, agent 1 and 2 can invite their friends 3, 4 and 5 and similarly, their friends can further invite agents 6, 7, 8 and 9. Finally, they reach a larger coalition including all agents in the forest. Suppose the characteristic function $v(E) = \mathbb{I}(|E| \geq 3 \land |\{3,7\}\cap E| = 0) + 2\mathbb{I}(|E| \geq 3 \land |\{3,7\}\cap E| \geq 1)$ for all $E\subseteq N$, where $\mathbb{I}(\cdot)$ is the indicator function.
% \[ v(E) = \begin{cases}
% 1 & \text{if } |E| \geq 3 \text{ and } |\{3,7\}\cap E| = 0 \\
% 2 & \text{if } |E| \geq 3 \text{ and } |\{3,7\}\cap E| \geq 1\\
% 0 & \text{otherwise}
% \end{cases} \qquad \text{ for every } E\subseteq N \]
% for every $E\subseteq N$. Clearly, the function $v$ is monotone.
\begin{figure}[htbp]
\center
\begin{tikzpicture}
\node[circle,inner sep=2.5pt,draw] (A) at (0,0) {$1$};
\node[circle,inner sep=2.5pt,draw] (C) at (0,-1) {$3$};
\node[circle,inner sep=2.5pt,draw] (B) at (2.5,0) {$2$};
\node[circle,inner sep=2.5pt,draw] (D) at (2.5,-1) {$4$};
\draw[->] (A) -> (C);
\draw[->] (B) -> (D);
\end{tikzpicture}
\caption{An example of the cooperative game in a forest.}\label{figure:1}
\end{figure}
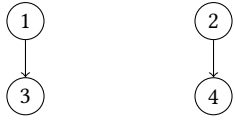
\end{example}

\subsection{Shapley Value does not Work}
The first question is what if we directly apply Shapley value~\cite{shapley1953value,winter2002shapley}, the classical solution for cooperative games, to our setting.
% Let us first apply the classical Shapley value~\cite{shapley1953value,winter2002shapley} in Example~\ref{example:1}.
Let $\mathcal{R}(S)$ denote the set of all orders $R$ of players in the coalition $S$. For an order $R$ in $\mathcal{R}(N)$, we denote by $B^{R,i}$ the set of players preceding $i$ in the order $R$. For a given characteristic function $v$ and an order $R$, the marginal contribution of player $i$ in $R$ is $C_i(v,R) = v(B^{R,i}\cup \{i\}) - v(B^{R,i})$. Then the classical Shapley value of game $v$ for $i$ is the expectation of $i$'s marginal contribution:
\[ \varphi_i(v) = \mathbb{E}_{\mathcal{U}_{\mathcal{R}(N)}}(C_i(v,\cdot)) \]
where $\mathcal{U}_{\mathcal{R}(N)}$ is a uniform distribution on $\mathcal{R}(N)$. We first show that Shapley value (on $J_{\mathcal{I}}(G(\theta'))$) does not work in our new setting.

\begin{proposition}
If Shapley value is applied as a reward distribution mechanism $\mathcal{M}$, then $\mathcal{M}$ is not diffusion incentive compatible.
\end{proposition}

\begin{proof}
We prove this proposition by presenting a counterexample. Consider the game in Example~\ref{example:1}. When applying Shapley value as the reward distribution mechanism, we get
\[
\begin{cases}
& \varphi_1 = \varphi_3 = 1/2, \\
& \varphi_2 = 0, \quad \varphi_4 = 1.
\end{cases}
\]
Notice that agents 1 and 3 have same contribution to other agents and they will share the reward of that contribution. However, if agent 1 does not invite agent 3, her Shapley value will be increased to $1 > 1/2$. Hence, directly applying Shapley value is not diffusion incentive compatible.
% Applying the Shapley value to the game in Example~\ref{example:1}, we can derive that $\phi_3 = \phi_7 = 217/504 $ and $\phi_i = 41/252$ for $i\neq 3, 7$. Notice that agent 3 and 7 has same contribution to other agents and they will share the reward of that contribution. Therefore, there is actually no incentive for agent 3 to invite agent 7 to join the coalition. If agent 3 does not invite agent 7 ($N'=N\setminus\{7\}$), her Shapley value will be increased to $7/8 > 217/504$. Hence, the Shapley value alone does not satisfy the DIC property.
\end{proof}

% \subsection{Unanimity Basis}
% We first introduce a useful notation of unanimity basis~\cite{harsanyi1959bargaining}, which plays an important role in analysing cooperative games. An unanimity basis of the cooperative game on player set $N$ is $\{u_E\mid E\subseteq N, E\neq \emptyset\}$, where an unanimity game $u_E$ is defined as $u_E(F) = \mathbb{I}(E\subseteq F)$, where $\mathbb{I}(\cdot)$ is the indicator function,
% % \[ u_E(F) = \begin{cases}
% % 1 & \text{if } E\subseteq F\\
% % 0 & \text{otherwise}
% % \end{cases} \]
% i.e., unit value is created when the player set $E$ exists in the coalition. Then any characteristic function $v$ can be expressed as $v(F) = \sum_{E\subseteq N} d^v(E)u_E(F)$,
% % \[ v(F) = \sum_{E\subseteq N} d^v(E)u_E(F) \]
% for all $F\subseteq N$, where $d^v(E) = \sum_{F\subseteq E} (-1)^{|E|-|F|} v(F)$ is called the dividend of coalition $E$ in $v$~\cite{harsanyi1959bargaining,chalkiadakis2011computational}.
% % and can be calculated as follows~\cite{harsanyi1959bargaining,chalkiadakis2011computational}.
% % \[ d^v(E) = \sum_{F\subseteq E} (-1)^{|E|-|F|} v(F) \]

% % With the unanimity basis, it will make the representation of our new solution easier and clearer.

\subsection{Permission Structure in Forests}
\label{sec:psf}
To tackle the failure of Shapley value, we recall an important concept called permission structure. % before our contributions.
% The last preliminary we introduce before our contributions is permission structure in forests.
Gilles et al.~\cite{gilles1992games,gilles1999cooperative} gave the permission restriction on cooperative games. They defined the permission structure on DAGs via the conjunctive and the disjunctive approach, respectively. We will extend the notions on DAGs in Section~\ref{sec:mixed}. Here we only utilize existing work on forests. Intuitively, a permission structure can represent how players get involved in the game by others' invitations. The formal definition is shown as below.
%via the conjunctive and disjunctive approach respectively. The axiomatic characterizations of the Shapley value via permission structure were given by~\cite{van1996axiomatizations,van1997axiomatization}. They defined permission structure on DAGs but relying on conjunctive or disjunctive approaches. We will extend the notions on DAGs in Section~\ref{sec:mixed} so here we only utilize existing work on forests.% A permission structure is defined as follows.

\begin{definition}
A permission structure on $N$ is an asymmetric mapping $P:N\rightarrow 2^N$, i.e., $j\in P(i)$ implies that $i\notin P(j)$.
\end{definition}

Here, we define $P(i)$ as the set of players who invited $i$ into the coalition, i.e., $P(i) = p_i' = \{ j\mid i\in \theta'_j \}$. In particular, in the forest model, every player except the initial players has a unique parent who invites her ($|p_i'|=1$). For instance, in Example~\ref{example:1}, $P(3) = \{1\}$ and $P(4) = \{2\}$. Then, we can define the autonomous coalition.% as follows.

\begin{definition}
A coalition $S\subseteq N$ is autonomous in a permission structure $P$ if for all $i\in S$, $P(i)\subseteq S$.
\end{definition}

A coalition $S\subseteq N$ is autonomous if and only if for each player $i\in S$, all her ancestors are also in $S$. The property of autonomous indicates whether a coalition has a chance to generate rewards by itself. For instance, in Example~\ref{example:1}, $\{1,3\}$ is autonomous while $\{4\}$ is not autonomous. Denote the collection of all autonomous coalitions in permission structure $P$ by $A_P$. Then for an arbitrary coalition $S\subseteq N$, we consider the largest autonomous part of it.

\begin{definition}
Let $P$ be a permission structure on $N$. Then the largest autonomous part of a coalition $S\subseteq N$ is defined by \[ \alpha(S) = \bigcup \{ T \mid T\subseteq S \text{ and } T\in A_{P} \}. \]
% \[ \alpha(E) = \bigcup \{ F \mid F\subseteq E \text{ and } F\in A_{S} \} \]
\end{definition}

Intuitively, $\alpha(S)$ is the largest subset of $S$ that is autonomous. In particular, in the forest model, let $G_S$ be the subgraph of the forest $G$ formed by players in $S$, then the largest autonomous part of the coalition $S$ is all the connected components of $G_S$ that contains at least one player in the initial set $\mathcal{I}$. For instance, in Example~\ref{example:1}, the largest autonomous part of set $\{1,3,4\}$ is $\{1,3\}$. %On the other hand, we consider the smallest autonomous coalition that precisely contains $E$.

% \begin{definition}
% Let $S$ be a permission structure on $N$. Then the authorizing coalition for a coalition $E\subseteq N$ is defined by $\lambda(E) = \bigcap \{ F \mid E\subseteq F \text{ and } F\in A_{S} \}$.
% % \[ \lambda(E) = \bigcap \{ F \mid E\subseteq F \text{ and } F\in A_{S} \} \]
% \end{definition}

% \textcolor{red}{add an example here}
% In a forest, the authorizing coalition of a coalition $E$ is obtained by taking union with all the players which are on the paths from the initial players to the players in $E$.

% \section{Diffusion Incentives in a Forest}

\subsection{Applying Permission Shapley Value}
\label{sec:apppsv}
Taking the structure of the diffusion forest into account, we see that some players need others' participation to create value. For instance, in Example~\ref{example:1}, player 4 can only be invited by player 2. Hence, without player 2, player 4 cannot provide her contribution to the coalition. This suggests applying Shapley value with a permission structure.

Taking the notations in Section~\ref{sec:psf}, with the restriction of the permission structure, we can map the characteristic function $v$ of the diffusion cooperative game to a projection $v^P$ on $P$ as 
\[v^P(S) = v(\alpha(S))\]
for all $S\subseteq N$~\cite{gilles1992games}.
% \[ v'(E) = v(\alpha(E)), \qquad \text{for every } E\subseteq N. \]
Intuitively, it means that the contribution given by a coalition $S$ is only from those players who can participate in the game in the coalition $S$.
% Alternatively, the mapping can also be expressed in terms of unanimity basis as~\cite{gilles1992games}

% \[ v'(F) = \sum_{E\in A_S} \left( \sum_{F:\lambda(E')=E} d^v(E') \right) u_E(F) \qquad \text{for every } F\subseteq N. \]

Define the permission Shapley value on characteristic function $v$ as $\varphi^P(v) = \varphi(v^P)$, i.e., the Shapley value on the game $v^P$. Applying the permission Shapley value in Example~\ref{example:1}, we have \[v^P(S) = \mathbb{I}(1\in S) + \mathbb{I}(\{2, 4\}\subseteq S)\] and the reward distributed to players will be
\[
\begin{cases}
& \varphi^P_1 = 1, \quad \varphi^P_3 = 0, \\
& \varphi^P_2 = \varphi^P_4 = 1/2.
\end{cases}
\]
%$\phi'_1 = 3/4$, $\phi'_2 = 4/15$, $\phi'_3 = 19/30$, $\phi'_4 = 1/15$, $\phi'_5 = 3/20$ and $\phi'_6 = \phi'_7 = \phi'_8 = \phi'_9 = 1/30$.

From the example we can see two intuitions of the permission Shapley value. Firstly, if a player has the same contribution as her inviter (e.g., player 3 and her inviter player 1 in Example~\ref{example:1}), then only the inviter will be rewarded for that contribution. This can be naturally obtained from the property of diffusion incentive compatibility. No matter how much reward the player shared with her inviter, the inviter will then have no incentives to invite the player. On the other hand, if a player invites another player who can create additional contribution (e.g., player 2 and player 4 in Example~\ref{example:1}), then the reward for the contribution will be equally shared with them.

% To see the intuition of the permission Shapley value here, we first
To see the intuition, we consider a special case where an additive assumption is applied, i.e., for each two disjoint subsets of players $S_1$, $S_2\subseteq N$ such that $S_1\cap S_2 = \emptyset$, we have $v(S_1\cup S_2) = v(S_1) + v(S_2)\geq 0$. Under this assumption, denote the depth of $i$ in the tree $i$ belongs to by $d_i$ (the depths of the roots are 0) and the subtree rooted by $i$ by $T_i$. Then the permission Shapley value of all player $i\in N$ is% $\phi'_i = \sum_{k\in T_i} \frac{v(\{k\})}{d_k+1}$.
\[ \varphi^P_i = \sum_{k\in T_i} \frac{v(\{k\})}{d_k+1} \qquad \text{for every } i\in N. \]
Intuitively speaking, the contribution of a player will be uniformly distributed as the reward along the invitation chain.

Now, we show that permission Shapley value is a desirable reward distribution mechanism that satisfies efficiency and diffusion incentive compatibility for diffusion cooperative games in forests.

\begin{theorem}
\label{thm:1}
For the monotone diffusion cooperative game in a forest, if the reward distribution mechanism $\mathcal{M}$ is the permission Shapley value, then $\mathcal{M}$ is efficient and DIC.
\end{theorem}

\begin{proof}
%(i) The efficiency is easily derived since it is a variation of the Shapley value.
(i) $\mathcal{M}$ is efficient since for all $\theta'\in \Theta(N)$
\[ \sum_{i\in N} \varphi_i^P(v) = \sum_{i\in N} \varphi_i(v^P) = v^P(J_{\mathcal{I}}(G(\theta'))) = v(J_{\mathcal{I}}(G(\theta'))). \]
% (i) The total value distributed by permission Shapley value in $v$ is same as the Shapley value in $v'$. Hence, we have $\sum_{i\in N} \phi'(v) = \sum_{i\in N} \phi(v') = v'(N) = v(N)$,
% \[ \sum_{i\in N} \phi'(v) = \sum_{i\in N} \phi(v') = v'(N) = v(N) \]
% which suggests the efficiency of permission Shapley value.

(ii) For the diffusion incentive compatibility, we will show that for each player $i$, her permission Shapley value is non-decreasing after she invites more players in the game. Consider a player set $X$ which cannot be informed of the game if $i$ does not invite some players. %(i.e., $i$ reports $\theta_i'\subseteq \theta_i$ and for all $j\in X$, $j\notin J_{\mathcal{I}}(G(\theta'))$) .
Let $P$ be the permission structure if $i$ does not invite these players and $P'$ be the permission structure if $i$ invites these players. Then
\begin{itemize}
    \item Before $i$ invites some players to let $X$ get involved in the game, the permission Shapley value of $i$ is
    \[ \varphi^P_i(v) = \varphi_i(v^P) = \frac{1}{|N\setminus X|!} \sum_{R\in \mathcal{R}(N\setminus X)} C_i(v^P, R). \]
    % \[ \phi'_i(v) = \phi_i(v') = \frac{1}{|N\setminus X|!} \sum_{R\in \mathcal{R}(N\setminus X)} C_i(v', R) \]
    \item If $i$ invites players such that $X$ then can be involved in the game, notice that for all $R\in \mathcal{R}(N\setminus X)$ and all $Y\subseteq X$, we have $v^P(B^{R,i}) = v^{P'}(B^{R,i}\cup Y)$ since $i$ is not in the coalition. Denote $C_Y(v^{P'}, R, i) = v^{P'}(B^{R,i}\cup Y \cup \{i\}) - v^{P'}(B^{R,i}\cup \{i\})$. Then the permission Shapley value of $i$ will become
    \begin{align*}
        \varphi^{P'}_i(&v) = \frac{1}{|N|!} \sum_{R\in \mathcal{R}(N)} C_i(v^{P'}, R) \\
        & = \frac{|X|!}{|N|!} \sum_{R\in \mathcal{R}(N\setminus X)} \left[ \binom{|N|}{|X|} C_i(v^P,R) + \sum_{Y\subseteq X} C_Y(v^{P'}, R, i) \right] \\ 
        & \geq \frac{|X|!}{|N|!} \sum_{R\in \mathcal{R}(N\setminus X)} \binom{|N|}{|X|} C_i(v^P,R) \\
        & = \frac{1}{(|N|-|X|)!} \sum_{R\in \mathcal{R}(N\setminus X)} C_i(v^P, R)
        = \varphi_i^P(v)
    \end{align*}
\end{itemize}
where the equality in the second line means the marginal contribution gained by $i$ when asserting the players in $X$ to all orders in $\mathcal{R}(N\setminus X)$. Therefore, the permission Shapley value is DIC.
\end{proof}

For structural fairness, we show that permission Shapley value satisfies 1-SF. Intuitively, it means that if a player $i$ invites $j$ to the game, then she will gain at least as much as the reward that is distributed to $j$.

\begin{theorem}\label{thm:sf}
For the monotone diffusion cooperative game in a forest, if the reward distribution mechanism $\mathcal{M}$ is the permission Shapley value, then $\mathcal{M}$ is 1-SF.
\end{theorem}

\begin{proof}
For all $\theta'\in \Theta$, consider $i,j\in J_{\mathcal{I}}(G(\theta'))$ with $j\in \theta_i'$. For all $S\subseteq N$ with $i\in S$, let $Q^i_S = (|N|-|S|)!(|S|-1)!/|N|!$, i.e., the probability of $B^{R,i} = S\setminus \{i\}$ in all orders $R\in \mathcal{R}(N)$. Since $R$ is sampled with uniform distribution, then for all $S\subseteq N$ with $i,j\in N$, $Q^i_S = Q^j_S$. Hence, we have
\begin{align}
    \varphi_i^P & = \sum_{S\ni i} Q^i_S \left[ v^P(S) - v^P(S\setminus\{i\}) \right] \notag \\
    & = \sum_{S\ni i, S\notni j} Q^i_S \left[ v^P(S) - v^P(S\setminus\{i\}) \right] \notag \\ &\quad + \sum_{S\ni i, S\ni j} Q^j_S \left[ v^P(S) - v^P(S\setminus\{i\}) \right] + \sum_{S\notni i, S\ni j} Q^j_S \cdot 0 \notag \\
    & \geq \sum_{S\ni i, S\ni j} Q^j_S \left[ v^P(S) - v^P(S\setminus\{i\}) \right] + \sum_{S\notni i, S\ni j} Q^j_S \cdot 0 \\
    & = \sum_{S\ni i, S\ni j} Q^j_S \left[ v^P(S) - v^P(S\setminus\{i\}) \right] \notag \\ &\quad + \sum_{S\notni i, S\ni j} Q^j_S \cdot \left[ v^P(S) - v^P(S\setminus\{j\}) \right] \\
    & \geq \sum_{S\ni i, S\ni j} Q^j_S \left[ v^P(S) - v^P(S\setminus\{j\}) \right] \notag \\ &\quad + \sum_{S\notni i, S\ni j} Q^j_S \cdot \left[ v^P(S) - v^P(S\setminus\{j\}) \right] \\
    & = \sum_{S\ni j} Q^j_S \left[ v^P(S) - v^P(S\setminus\{j\}) \right] = \varphi_j^P \notag
\end{align}
where the Inequality (1) is satisfied since the game is monotone; the Equality (2) is satisfied since for all $S\subseteq N$ with $i\notin S$ and $j\in S$, $j\notin \alpha(S)$ and then $v^P(S) - v^P(S\setminus\{j\}) = v(\alpha(S)) - v(\alpha(S)) = 0$; the Inequality (3) is satisfied since for all $S\subseteq N$ with $i,j\in S$, $\alpha(S\setminus\{i\}) = \alpha(S\setminus\{i, j\})$ and then $v^P(S) - v^P(S\setminus\{i\}) = v^P(S) - v^P(S\setminus\{i, j\}) \geq v^P(S) - v^P(S\setminus\{j\})$.

Therefore, the permission Shapley value is 1-SF.
\end{proof}

\subsection{Using Weights to Further Utilize the Structure}
Now we look back at the game illustrated in Example~\ref{example:1}. The permission Shapley value suggests an equal share between players 2 and 4. However, in this example, player 2 has diffusion incentives if sharing any amount from player 4's contribution. Simply applying permission Shapley value cannot tell the differences between these two players. Alternatively speaking, we need a method that can tune the parameter $\gamma$ in structural fairness. Therefore, this suggests that weights can be introduced to distinguish the differences among these players. Kalai and Samet~\cite{kalai1987weighted} introduced weights to the Shapley value as an alternative solution to cooperative games. Radzaik~\cite{radzik2012new} further discussed the variants and properties of the weighted Shapley value and Dragan~\cite{dragan2008computation} provided a computation method for weighted Shapley value. Usually, the weighted Shapley value can be defined as:
\[ \varphi^\omega_i(v) = \mathbb{E}_{\mathcal{D}(\omega)}(C_i(v,\cdot)) \]
where $\omega = (\omega(1), \omega(2), \dots, \omega(n)) \in \mathbb{R}^{N}_+$ are the weights assigned to players and $\mathcal{D}(\omega)$ is a distribution on $\mathcal{R}(N)$ based on $\omega$.

To compute $\mathcal{D}(\omega)$, consider an order $R\in\mathcal{R}(N) = (i_1, i_2, \dots, i_n)$, define $\omega_R = \prod_{k=1}^m$ $\left( \omega(i_k) / \sum_{p=1}^k \omega(i_p) \right)$. This can be interpreted as the probability of sampling the order $R$ by agents' weights, e.g., sampling last player as $i_n$ has probability $\omega(i_n) / (\omega(i_1) + \cdots + \omega(i_{n}))$ and sampling previous player as $i_{n-1}$ in the remaining players has probability $\omega(i_{n-1})/(\omega(i_1)+\cdots+\omega(i_{n-1}))$. Finally, in $\mathcal{D}(\omega)$, the probability of selecting order $R$ is $\omega_R$~\cite{kalai1987weighted}. Table~\ref{tab:weights} shows an example of the weight assignments.

\begin{table}[!htbp]
    \centering
    \begin{tabular}{|c|c|c|}
        \hline
        order $R$ & $\omega_R$ ($\omega(i) = 1$ for all $i$) & $\omega_R$ ($\omega(1 \text{ or } 2) = 1, \omega(3) = 2$) \\
        \hline
        (1, 2, 3) & 1/6 & 1/4 \\
        (1, 3, 2) & 1/6 & 1/6 \\
        (2, 1, 3) & 1/6 & 1/4 \\
        (2, 3, 1) & 1/6 & 1/6 \\
        (3, 1, 2) & 1/6 & 1/12 \\
        (3, 2, 1) & 1/6 & 1/12 \\
        \hline
    \end{tabular}
    \caption{An example of computing the $\mathcal{D}(\omega)$ from weights $\omega$. We can see when $\omega(3)$ is larger, player 3 has more chance to appear the the later positions.}
    \label{tab:weights}
\end{table}

Note that when $\omega = 1^{|N|}$, the weighted Shapley value becomes the classical Shapley value. In our setting, we can also assign weights to players. Intuitively, the permission structure shows some kinds of ``external" relations of the players: how players are connected in a social network; while the weights show some ``internal" relations: which player takes a more important role in a coalition. Thus, these two solution concepts are of different classes. In our reward distribution mechanism, we may want to consider not only the ``external" structures but also ``internal" relations between players involved. Moreover, from the perspective of fairness, the weights will decide how much a player $i$ can be rewarded by inviting her neighbours, i.e., the parameter $\gamma$ in structural fairness. Therefore, we introduce a new idea as weighted permission Shapley value.

\begin{definition}
For a cooperative game $v$ on the player set $N$, given a permission structure $P$ and a weight $\omega\in \mathbb{R}^{|N|}_+$, the weighted permission Shapley value for a player $i\in N$ is:
\[ \varphi_i^{\omega, P}(v) = \varphi_i^{\omega} (v^P). \]
% \[ \phi^{\omega,S}_i(v) = \sum_{E\in A_S} \left( \sum_{F:\lambda(F)=E} d^v(F) \right) \phi^{\omega}_i(u_E) \]
\end{definition}

To apply weighted permission Shapley value to a diffusion cooperative game, we still set $P(i)=p'_i$ and then we need a weight function $\omega(i)$ to set weights to each player $i$. As an example, let the weight function be $\omega(i) = d_i + 1$,
% \[\omega(i) = \begin{cases}
% 1 & \text{if } d_i = 0 \\
% d_i & \text{otherwise}
% \end{cases}\]
where $d_i$ is the depth of player $i$ in the tree $i$ belongs to\footnote{For players who are not in the set $J_{\mathcal{I}}(G(\theta'))$, they can be assigned arbitrary positive weight. We will not specify the conditions for these players later unless necessary.}. Then applying the weighted permission Shapley value in Example~\ref{example:1}, the rewards distributed to players are: %$\phi^{\omega, S}_1 = 157/210$, $\phi^{\omega, S}_2 = 11/42$, $\phi^{\omega, S}_3 = 247/420$, $\phi^{\omega, S}_4 = 43/420$, $\phi^{\omega, S}_5 = 23/210$, $\phi^{\omega, S}_6 = \phi^{\omega, S}_7 = 3/70$ and $\phi^{\omega, S}_8 = \phi^{\omega, S}_9 = 11/210$.
\[
\begin{cases}
& \varphi^{\omega,P}_1 = 1, \quad \varphi^{\omega,P}_3 = 0, \\
& \varphi^{\omega,P}_2 = 1/3 , \quad \varphi^{\omega,P}_4 = 2/3.
\end{cases}
\]
% We can see it more clearly in Example~\ref{example:wsv}.

% \begin{example}
% \label{example:wsv}
% For the diffusion cooperative game shown in Example~\ref{example:1}, let the weight function $\omega$ for players be
% \[\omega(i) = \begin{cases}
% 1 & \text{if } d_i = 0 \\
% d_i & \text{otherwise}
% \end{cases}\]
% where $d_i$ is the depth of player $i$ in the tree $i$ belongs to. Then applying the weighted permission Shapley value, the rewards distributed to players are
% \[ \begin{cases}
% \phi^{\omega, S}_1 = 157/210, \quad \phi^{\omega, S}_2 = 11/42 \\
% \phi^{\omega, S}_3 = 247/420, \quad \phi^{\omega, S}_4 = 43/420, \quad \phi^{\omega, S}_5 = 23/210 \\
% \phi^{\omega, S}_6 = \phi^{\omega, S}_7 = 3/70 \\
% \phi^{\omega, S}_8 = \phi^{\omega, S}_9 = 11/210 \\
% \end{cases} \]
% \end{example}

From the example we can see that we make a difference between players 2 and 3's rewards. Again, consider the special case for the additive characteristic function $v$. Let $T_i$ be the subtree rooted by $i$. Then the weighted permission Shapley value with weight $\omega: \omega(i) = f(d_i)$ for all player $i\in N$ is
% $\phi^{\omega,S}_i = \sum_{k\in T_i} \frac{f(d_i)}{\sum_{j=0}^{d_k}f(j)} v(\{k\})$.
\[ \varphi^{\omega,P}_i = \sum_{k\in T_i} \frac{f(d_i)}{\sum_{j=0}^{d_k}f(j)} v(\{k\}).  \]
Intuitively speaking, the reward will be distributed along the invitation chain according to the ratio of the weights rather than uniformly divided.

If we set weights as $1^{|\mathcal{N}|}$, the weighted permission Shapley value will become normal permission Shapley value. Hence, the weighted permission Shapley value is a more general class of mechanisms and we will show that if we set weights properly, it is also a desirable solution to diffusion cooperative game that satisfies efficiency and diffusion incentive compatibility.

\begin{theorem}
\label{thm:tree}
For the monotone diffusion cooperative game in a forest, if the reward distribution mechanism $\mathcal{M}$ is the weighted permission Shapley value with weight function $\omega(i)=f(d_i)$ for all player $i$, which only depends on her distance to initial players, then $\mathcal{M}$ is efficient and DIC.
\end{theorem}

\begin{proof}
(i) $\mathcal{M}$ is efficient since for all $\theta'\in \Theta(N)$, 
\[ \sum_i \varphi_i^{\omega, P}(v) = \sum_i \varphi_i^{\omega}(v^P) = v^P(J_{\mathcal{I}}(G(\theta'))) = v(J_{\mathcal{I}}(G(\theta'))). \]
% \[ \sum_i \phi_i^{\omega, S}(v) = \sum_i \phi_i^{\omega}(v') = v'(N) = v(N) \]

(ii) For the property of DIC, suppose $X$ is the player set which cannot be informed of the game if $i$ does not invite some players. Let $P$ be the permission structure if $i$ does not invite these players and $P'$ be the permission structure if $i$ invites these players. Then, 
%For any order $R\in\mathcal{R}(N\setminus X) = (i_1, i_2, \dots, i_m)$, define $\omega_R = \prod_{k=1}^m$ $\left( \omega(i_k) / \sum_{p=1}^k \omega(i_p) \right)$. This can be interpreted as the probability of sampling the order $R$ by agents' weights (e.g., sampling last player as $i_m$ has probability $1/(\omega(i_1)+\cdots+\omega(i_{m-1}))$ and sampling previous player as $i_{m-1}$ in the remaining players has probability $1/(\omega(i_1)+\cdots+\omega(i_{m-2}))$).
% Then referring to the previous work~\cite{kalai1987weighted},
the weighted permission Shapley value before $i$ let $X$ get involved in the game is
\[\varphi_i^{\omega}(v^P)=\sum_{R\in\mathcal{R}(N\setminus X)}\omega_R C_i(v^P,R) \left/ \sum_{R\in\mathcal{R}(N\setminus X)} \omega_R \right. .\]
Consider an order $R_j^p=(i_1,\dots,i_{p-1},j,i_p,\dots,i_m)$, which inserts player $j$ at the position $p$ in $R$. Then from the definition we can derive that $\omega_R = \sum_{p=1}^{m+1} \omega_{R_j^p}$ if for all $k$, $\omega(i_k)$ will not change after $j$ joins in. More generally, for any additional player set $X$, if for all $k$, $\omega(i_k)$ will not change after $X$ joins in, we have $\sum_{R'\in R_X} \omega_{R'} = \omega_R$, where $R_X$ is the set of all possible orders that insert all players in $X$ into the order $R$. Then if $i$ invites players such that $X$ can be involved in the game, since the weight function $\omega(i) = f(d_i)$ only depends on $d_i$, for all player $i\in N\setminus X$, $\omega(i)$ will not change. Hence, the weighted permission Shapley value of $i$ becomes
\begin{align*}
    \varphi^{\omega}_i&(v^{P'}) = \frac{1}{\sum_{R\in \mathcal{R}(N)}\omega_R} \sum_{R\in \mathcal{R}(N)} \omega_R C_i(v^{P'}, R) \\
    & = \frac{1}{\sum_{R\in \mathcal{R}(N\setminus X)}\omega_R} \sum_{R\in \mathcal{R}(N\setminus X)} \sum_{R'\in R_X} \omega_{R'} C_i(v^P, R') \\
    & \geq \frac{1}{\sum_{R\in \mathcal{R}(N\setminus X)}\omega_R} \sum_{R\in \mathcal{R}(N\setminus X)} \sum_{R'\in R_X} \omega_{R'} C_i(v^P, R) \\
    & = \frac{1}{\sum_{R\in \mathcal{R}(N\setminus X)}\omega_R} \sum_{R\in \mathcal{R}(N\setminus X)} \omega_R C_i(v^P, R)
    % & = \frac{1}{\sum_{R\in \mathcal{R}(N\setminus X)}\omega_R} \sum_{R\in \mathcal{R}(N\setminus X)} \omega_R C_i(v', R)
    = \varphi_i^\omega(v^P).
\end{align*}

Therefore, $\mathcal{M}$ is DIC.
% Therefore, the weighted permission Shapley value with weight function $\omega(i) = f(d_i)$ is DIC.
% (ii) For the diffusion compatibility, the total contribution of player $i$ without inviting player $j$ is a weighted sum of contributions in all orders without $j$. If $i$ does invite $j$, for each order where $j$ comes after $i$, the marginal contribution of $i$ will not change; for each order $j$ comes before $i$, the marginal contribution of $i$ is increased by either $v_j$ or $0$. Suppose $R$ is an order in $\mathcal{R}(\mathcal{N}\setminus\{j\})$, and $R_j^p$ is the order with $j$ inserts at position $p$ in $R$. Then
% \[ \sum_{p} w'_{R_j^p} C_i(v',R_j^p) \geq \sum_{p} w'_{R_j^p} C_i(v',R) = w'_{R} C_i(v',R) \]

% Here, $w'$ is the weight of the orders, which can be computed from weights $\omega$~\cite{kalai1987weighted} and $\sum_{p} w'_{R_j^p} = w'_{R}$ since all the players' weight do not change (the weight only depends on the depth) except $j$'s. Note that the left-hand term is player $i$'s reward after inviting $j$ while the right-hand term is $i$'s reward without inviting $j$. Thus, the reward of $i$ when she invites $j$ is no less than that when she does not invite $j$.

% Therefore, the mechanism from our proposed family is DIC.
\end{proof}

Moreover, by introducing weights, we can make the structural fairness more tunable to customize the requirements in different scenes.

\begin{theorem}
For the monotone diffusion cooperative game in a forest, if the reward distribution mechanism $\mathcal{M}$ is the weighted permission Shapley value with weight function $\omega(i)$ that satisfies for all $\theta'\in \Theta$ and for all $i,j \in J_{\mathcal{I}}(G(\theta'))$ with $j\in \theta'_i$, $\omega(i)/\omega(j) \geq \gamma$, then $\mathcal{M}$ is $\gamma$-SF.
\end{theorem}

\begin{proof}
For all $\theta'\in \Theta$, consider $i,j\in J_{\mathcal{I}}(G(\theta'))$ with $j\in \theta_i'$. For all $S\subseteq N$ with $i\in S$, let $Q^i_S$ be the probability of $B^{R,i} = S\setminus \{i\}$ in all orders $R\in \mathcal{R}(N)$. Since $R$ is sampled with distribution $\mathcal{D}(\omega)$, then for all $S\subseteq N$ with $i,j\in N$, we have
\[ \frac{Q^i_S}{Q^j_S} = \frac{\omega(i)}{\omega(j)}. \]
Hence, $Q^i_S = \frac{\omega(i)}{\omega(j)}Q^j_S$ and we have
\begin{align}
    \varphi_i^{\omega, P} & = \sum_{S\ni i} Q^i_S \left[ v^P(S) - v^P(S\setminus\{i\}) \right] \notag \\
    & = \sum_{S\ni i, S\notni j} Q^i_S \left[ v^P(S) - v^P(S\setminus\{i\}) \right] \notag \\ &\quad + \frac{\omega(i)}{\omega(j)} \left\{ \sum_{S\ni i, S\ni j} Q^j_S \left[ v^P(S) - v^P(S\setminus\{i\}) \right] + \sum_{S\notni i, S\ni j} Q^j_S \cdot 0 \right\} \notag \\
    & \geq \frac{\omega(i)}{\omega(j)} \left\{ \sum_{S\ni i, S\ni j} Q^j_S \left[ v^P(S) - v^P(S\setminus\{i\}) \right] + \sum_{S\notni i, S\ni j} Q^j_S \cdot 0 \right\} \notag \\
    % & = \sum_{S\ni i, S\ni j} Q^j_S \left[ v^P(S) - v^P(S\setminus\{i\}) \right] \notag \\ &\quad + \sum_{S\notni i, S\ni j} Q^j_S \cdot \left[ v^P(S) - v^P(S\setminus\{j\}) \right] \\
    & \geq \frac{\omega(i)}{\omega(j)} \left\{ \sum_{S\ni i, S\ni j} Q^j_S \left[ v^P(S) - v^P(S\setminus\{j\}) \right] \right. \notag \\ &\quad + \left. \sum_{S\notni i, S\ni j} Q^j_S \cdot \left[ v^P(S) - v^P(S\setminus\{j\}) \right] \right\} \\
    & = \frac{\omega(i)}{\omega(j)} \left\{ \sum_{S\ni j} Q^j_S \left[ v^P(S) - v^P(S\setminus\{j\}) \right] \right\} \notag \\
    & = \frac{\omega(i)}{\omega(j)} \varphi_j^{\omega, P} \geq \gamma \varphi_j^{\omega, P} \notag
\end{align}
where the Inequality (4) is satisfied according to the same reason for Equality (2) and Inequality (3) in Theorem~\ref{thm:sf}.

Therefore, the $\mathcal{M}$ is $\gamma$-SF.
\end{proof}

Intuitively, the parameter of the structural fairness is determined by $\min_{j\in \theta'_i} \omega(i)/\omega(j)$. For example, if $\omega(i) = d_i + 1$, then the corresponding weighted permission Shapley value is $1/2$-SF.

\section{The Only Solution to Query Network}
A classic problem that can be modelled as a diffusion cooperative game is the query incentive network~\cite{kleinberg2005query}, where a requester tries to find an answer to a specific problem by diffusing the request in the network. A solution is given by the winning team from MIT in the DARPA red balloon challenge~\cite{pickard2011time}. In the challenge, each team needed to find positions of the red balloons to obtain rewards. The solution proposed by the winning team is that they promised half of the reward for the first person who found it and one-fourth of the reward to the person who invited the finder and so on. The requester (initial players) will get the remaining. An example is shown in Figure~\ref{figure:redballoon}.
% Notice that this solution cannot be represented by permission Shapley value since it cannot distinguish the differences among a player's ancestors.

\begin{figure}[htbp]
\center
\begin{tikzpicture}
\node[circle,inner sep=2.5pt,draw] (R) at (0,0) {$1$};
\node[blue] (R2) at (-.8,0) {\$1000};
\node[circle,inner sep=2.5pt,draw] (A) at (-1,-1) {$3$};
\node[blue] (A2) at (-1.8,-1) {\$1000};
\node[circle,inner sep=2.5pt,draw] (B) at (1,-1) {$4$};
\node[circle,inner sep=2.5pt,red,draw] (C) at (-1.5,-2) {$6$};
\node[blue] (C2) at (-2.3,-2) {\$2000};
\node[circle,inner sep=2.5pt,draw] (D) at (-0.5,-2) {$7$};
\node[circle,inner sep=2.5pt,draw] (E) at (3.5,0) {$2$};
\node[blue] (E2) at (2.7,0) {\$1000};
\node[circle,inner sep=2.5pt,draw] (F) at (3.5,-1) {$5$};
\node[blue] (F2) at (2.7,-1) {\$1000};
\node[circle,inner sep=2.5pt,red,draw] (G) at (3,-2) {$8$};
\node[blue] (G2) at (2.2,-2) {\$2000};
\node[circle,inner sep=2.5pt,draw] (H) at (4,-2) {$9$};
\draw[->] (R) -> (A);
\draw[->] (R) -> (B);
\draw[->] (A) -> (C);
\draw[->] (A) -> (D);
\draw[->] (E) -> (F);
\draw[->] (F) -> (G);
\draw[->] (F) -> (H);
\end{tikzpicture}
\caption{An example of the solution given by winning team in DARPA challenge. Players 1 and 2 are the initial team members. Players 6 and 8 are those who find the balloon.}\label{figure:redballoon}
\end{figure}
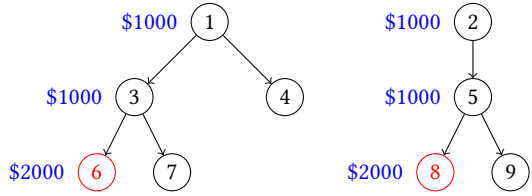

% Now we look back at the query incentive network.
We can model it as a diffusion cooperative game with an additive characteristic function where only one agent (the answer holder) can contribute the utility. Without loss of generality, we assume there is only one initial player as the requester and only one player can provide the answer and the answer will bring one unit value (for the game in the example shown in Figure~\ref{figure:redballoon}, we can seperate it as two games and add the two solutions up). More precisely, in the corresponding diffusion cooperative game to the query network, we set $\mathcal{I} = \{1\}$ and for any $S\subseteq N$, $v(S) = 1$ if and only if the answer holder $j\in S$. In general, a solution to the query network is a reward distribution $x(i)$ for all players $i$ along the path from the requester to answer provider. We require that the reward distribution $x(i)$ satisfies the following properties. %An anonymous, strongly individually rational (SIR) and efficient solution $x$ satisfies that for all $i$ on the path from player 1 to $j$, $x(i)$ only depends on $d_i$ and $d_j$, $x(i)>0$ and $\sum_i x(i) = 1$.
\begin{definition}
A reward distribution $x(i)$ for all players $i$ along the path from the requester 1 to answer provider $j$ in the query network is
\begin{itemize}
    \item \emph{anonymous} if $x(i)$ only depends on $d_i$ and $d_j$ (the distances from player 1 to $i$ and $j$, which  indicates $i$'s position);
    \item \emph{strongly individually rational} (SIR) if $x(i)>0$ for all $i$ from 1 to $j$;
    \item \emph{efficient} if $\sum_i x(i) = 1$.
\end{itemize}
\end{definition}

For example, the solution given by the DARPA winning team can be described as $x(i) = 1/2^{d_j + 1 - d_i}$ with $i>1$ and $x(1) = 1/2^{d_j}$ for all $i$ on the path from 1 to $j$. We show that all the solution concepts can be mapped to a set of weighted permission Shapley values. In other words, the set of weighted permission Shapley value is the only satifiable solution to the query network.

\begin{theorem}
A solution to a query network is anonymous, strongly individually rational and efficient if and only if it is a weighted permission Shapley value with $\omega(i) = f(d_i, d_j)$, where $j$ is the answer provider.
\end{theorem}

\begin{proof}
``$\Rightarrow$": suppose $x(i)$ is an anonymous, SIR and efficient solution to the query network. Construct a weighted permission Shapley value %for the corresponding diffusion cooperative game
with $\omega(i) = x(i)$ for all $i$ on the path from agent 1 to $j$ and $\omega(i) = 1$ for other players. Then,% we have
\[ \varphi^{\omega,P}_i = \sum_{k\in T_i} \frac{\omega(i)}{\sum_{l=0}^{d_k}\omega(l)} v(\{k\}) = \frac{x(i)}{\sum_l x(l)} = x(i)  \]
for all $i$ on the path from agent 1 to $j$ and otherwise $\varphi^{\omega,P}_i = 0$.

``$\Leftarrow$": consider a weighted permission Shapley value with $\omega(i) = f(d_i, d_j)$. $\varphi^{\omega,P}_i = 0$ if $i$ is not an ancestor of $j$. For all $i$ on the path from 1 to $j$, we have
\[ \varphi^{\omega,P}_i = \frac{f(d_i,d_j)}{\sum_{k\in\text{ path from agent 1 to } j} f(d_k, d_j)} = \frac{f(d_i,d_j)}{\sum_{k=0}^{d_j} f(k, d_j)} > 0 \]
which only depends on $d_i$ and $d_j$. Finally, the efficiency holds since $v(N) = 1$.
\end{proof}

Again, take the solution of DARPA winning team as an example. Consider the weighted permission Shapley value with $\omega(i) = \max\{1, 2^{d_i - 1}\}$, and we have %$\phi_i^{\omega,S} = \frac{2^{d_i-1}}{1+\sum_{k=1}^{d_j}2^{k-1}} = \frac{2^{d_i-1}}{2^{d_j}} = 1/2^{d_j+1-d_i}$,
\[ \varphi_i^{\omega,P} = \frac{2^{d_i-1}}{1+\sum_{k=1}^{d_j}2^{k-1}} = \frac{2^{d_i-1}}{2^{d_j}} = 1/2^{d_j+1-d_i} \]
% \[ \omega(i) = \begin{cases}
% 1 & \text{if } d_i = 0 \\
% 2^{d_i - 1} & \text{otherwise} \\
% \end{cases} \]
% which can be easily verified the equivalence to it by checking
% \[ \phi_i^{\omega,S} = \frac{2^{d_i-1}}{1+\sum_{k=1}^{d_j}2^{k-1}} = \frac{2^{d_i-1}}{2^{d_j}} = 1/2^{d_j+1-d_i} \]
for all $i$ on the path from agent 1 to $j$ with $i>1$, $\varphi_1^{\omega,P} = 1/2^{d_j}$ and $\varphi_i^{\omega,P} = 0$ otherwise, which is equivalent to the solution of DARPA winning team. Moreover, in this example, $\omega(i)$ only depends on $d_i$, so that we know the solution of DARPA winning team is diffusion incentive compatible according to Theorem~\ref{thm:tree}. %an equivalent solution.% in DARPA challenge.

\section{From Forest to DAG}
\label{section:dag}
In this section, we extend our result in the setting of forest to a general DAG model. An instance of a cooperative game in a DAG %where an player may be invited by several other players
is shown in Example~\ref{example:dag} below.

\begin{example}
\label{example:dag}
Consider the case illustrated in Figure~\ref{figure:dag}, where %the initial coalition is
$\mathcal{I} = \{1, 2\}$. Agent $1$ asks her friends $3$ and $5$, and $2$ asks $4$ and $5$. Then $3$, $4$ and $5$ further ask their friends and so on. %Finally, the structure constructed via diffusion becomes a DAG.
Suppose the player $5$ will join in if $2$ invites her or $1$ and $3$ both invite her and the player $7$ will join in if $4$ invites her or $5$ and $6$ both invite her. Suppose the characteristic function $v$ is defined as for every $S\subseteq N$, %$v(E) = \mathbb{I}(\{1, 2\}\cap E \neq \emptyset\land 7\notin E) + 2\mathbb{I}(7\in E)$ for all $E\subseteq N$ ($\mathbb{I}(\cdot)$ is the indicator).%, where $\mathbb{I}(\cdot)$ is the indicator function.
\[ v(S) = \begin{cases}
2 & \text{if } 7\in S; \\
1 & \text{if } \{1, 2\}\cap S \neq \emptyset, 7\notin S; \\
0 & \text{otherwise.}
\end{cases} \] %\qquad \text{ for every } S\subseteq N \]

\begin{figure}[htbp]
\center
\begin{tikzpicture}
\node[circle,inner sep=2.5pt,draw] (C) at (0,0) {$5$};
\node[circle,inner sep=2.5pt,fill=green!20,draw] (R) at (150:1.5) {$1$};
\node[circle,inner sep=2.5pt,fill=green!20,draw] (T) at (90:1.5) {$2$};
\node[circle,inner sep=2.5pt,draw] (A) at (-150:1.5) {$3$};
\node[circle,inner sep=2.5pt,draw] (B) at (30:1.5) {$4$};
\node[circle,inner sep=2.5pt,draw] (D) at (-90:1.5) {$6$};
\node[circle,inner sep=2.5pt,draw] (E) at (-30:1.5) {$7$};
\draw[->] (R) -> (A);
\draw[->] (T) -> (B);
\draw[->] (R) -> (C);
\draw[->] (T) -> (C);
\draw[->] (A) -> (C);
\draw[->] (A) -> (D);
\draw[->] (B) -> (E);
\draw[->] (C) -> (E);
\draw[->] (D) -> (E);
\end{tikzpicture}
\caption{An example of a cooperative game in DAG. The green nodes are initial players.}\label{figure:dag}
\end{figure}
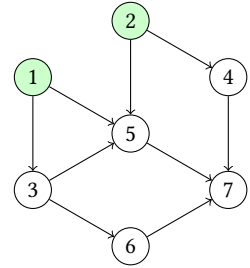
% Suppose the characteristic function $v$ will be:
% \[ v(E) = \begin{cases}
% 2 & \text{if } 7\in E \\
% 1 & \text{if } \{1, 2\}\cap E \neq \emptyset, 7\notin E \\
% 0 & \text{otherwise}
% \end{cases} \qquad \text{ for every } E\subseteq N \]
% for every $E\subseteq N$.
\end{example}

\subsection{Permission Structure with Mixed Approach}
\label{sec:mixed}
Note that there is no existing approach of permission structure that can handle the case in Example~\ref{example:dag}. Gilles et al.~\cite{gilles1992games,gilles1999cooperative} considered the cases where each player has to get permissions either from all or at least one of her superiors. %which is called the conjunctive and disjunctive approach respectively.
Here we consider a more general case where each player can get permission from a partial subset of her superiors. %Hence to generalize the result to DAG, we first 
Hence, we propose the permission structure with mixed approach. %The axiomatic characterizations of the Shapley value via permission structure were given by~\cite{van1996axiomatizations,van1997axiomatization}.

A permission structure with mixed approach $\varrho$ on $N$ is a pair $(P, \Psi)$ where $P$ is a mapping $N\rightarrow 2^N$. The mapping $P$ is asymmetric, i.e., for any pair $i$, $j\in N$, $j\in P(i)$ implies that $i\notin P(j)$ and  $j$ is called a superior of $i$. Define $P^{-1}(i) = \{ j\in N \mid i\in P(j) \}$ as the set of $i$'s successors. Notice that $P(i) = \emptyset$ if $i\in \mathcal{I}$. The set $\Psi = \{ \psi_i \mid i\in N \}$ consists of players' satisfiable expressions. For a coalition $S \subseteq N$, the expression set $L_S$ is recursively defined:% as following:
\begin{enumerate}
\item $\xi_i\in L_S$ for any $i\in S$;
\item $a\vee b\in L_S$ for any $a$, $b\in L_S$;
\item $a\wedge b\in L_S$ for any $a$, $b\in L_S$.
\end{enumerate}
Given an expression $\psi \in L_S$ and a coalition $T\subseteq N$, the evaluation $\psi(T)$ is the boolean result of $\psi$ when $\xi_i = 1$ if $i\in T$ and $\xi_i = 0$ otherwise for all $i\in S$. Then for $i\in N$, her satisfiable expression $\psi_i\in L_{P(i)}$, indicates how her superiors hold the authority of permission: only when $\psi_i(T)$ is true, $i$ can get the permission to create value in $T$. Specially, if $i\in \mathcal{I}$, $\psi_i$ is always true since $i$ does not need any others' permission.
For instance, in Example~\ref{example:dag}, $\psi_5 = \xi_2\vee (\xi_1\wedge \xi_3)$ and $\psi_7 = \xi_4\vee (\xi_5\wedge \xi_6)$. With the generalized permission structure, an autonomous coalition now can be defined as follows.

\begin{definition}
A coalition $S\subseteq N$ is autonomous in the permission structure $\varrho = (P, \Psi)$ if for all $i\in S$, $\psi_i(S) = 1$.
\end{definition}

Denote the set of all autonomous coalitions in $\varrho$ by $A_{\varrho}$. We can observe several properties of $A_{\varrho}$ as follows. % It can be observed that $\{\emptyset, N\}\subseteq A_{\varrho}$ ($\varphi(N) = 1$ since all variables are true). Most importantly, $A_{\varrho}$ is closed when taking union, i.e., for all $E,F\in A_{\varrho}$, $E\cup F\in A_{\varrho}$ (if $i\in E$, $\varphi_i(E) = 1$ implies $\varphi_i(E\cup F) = 1$ since more variables get true; so does $i\in F$).

\begin{lemma}
\label{lemma:1}
Let $\varrho$ be a permission structure on $N$, then (i) $\emptyset\in A_{\varrho}$, (ii) $N\in A_{\varrho}$ and (iii) for all $S,T\in A_{\varrho}$, $S\cup T\in A_{\varrho}$.
\end{lemma}

\begin{proof}
(i) Since there is no $i\in \emptyset$, then $\emptyset\in A_{\varrho}$. (ii) for all $\psi\in L_S$, $S\subseteq N$, $\psi(N) = 1$ since all variables are true. Thus, $N\in A_{\varrho}$. (iii) for all $i\in S\cup T$, if $i\in S$, $\psi_i(S) = 1$ implies $\psi_i(S\cup T) = 1$ since more variables get true; if $i\in T$, similarly we have $\psi_i(S\cup T) = 1$. Thus, $S\cup T \in A_{\varrho}$.
\end{proof}

% Since $A_{\varrho}$ is closed when taking union,
% Then we can give the definition of an autonomous part of a coalition.% as follows.
Then we can give the definition of the largest autonomous part of a coalition.

\begin{definition}
Let $\varrho$ be a permission structure on $N$. Then the largest autonomous part of a coalition $S\subseteq N$ is defined by
\[ \alpha(S) = \bigcup \{ T \mid T\subseteq S \text{ and } T\in A_{\varrho} \}. \]
% \[ \alpha(E) = \bigcup \{ F \mid F\subseteq E \text{ and } F\in A_{\varrho} \} \]
\end{definition}

Intuitively, $\alpha(S)$ is the largest autonomous sub-coalition of $S$, which suggests that for any player $i\in S\backslash \alpha(S)$, she cannot create value in coalition $S$. %On the other hand, we consider the minimal autonomous coalitions that contains $E$.

% \begin{definition}
% Let $\varrho$ be a permission structure on $N$. Then $F\subseteq N$ is an authorizing coalition of a coalition $E\subseteq N$ if (i) $F\in A_{\varrho}$, (ii) $E\subseteq F$ and (iii) for all $F'\in A_{\varrho}$, $E\subseteq F'\subseteq F \Rightarrow F' = F$. Denote the set of all authorizing coalitions of $E$ by $\Lambda(E)$. Furthermore, we consider a larger set defined as $\Lambda^*(E)=\{F = \bigcup_{i=1}^m F_m \mid F_i\in \Lambda(E), 1\leq i \leq m, m\in \mathbb{Z}^+ \}$. Clearly, $\Lambda^*(E) \subseteq A_{\varrho}$.
% \end{definition}

% Thus a special set of coalitions that contains $E$ is considered as $\Lambda^*(E)=\{F = \bigcup_{i=1}^m F_m \mid F_i\in \Lambda(E), 1\leq i \leq m, m\in \mathbb{Z}^+ \}$. Clearly, $\Lambda^*(E) \subseteq A_{\varrho}$.

%\subsection{Permission Restriction}
Similar to Section~\ref{sec:apppsv}, we can map a characteristic function $v$ to a projection $v^{\varrho}$ on $\varrho$, where $v^{\varrho}(S) = v(\alpha(S))$, for every coalition $S\subseteq N$.

\subsection{Weighted Shapley Value on Permission Structure}
Finally, we introduce weighted Shapley value with mixed permission structure as a class of mechanisms for diffusion cooperative game on DAGs.

\begin{definition}
For a cooperative game $v$ on the player set $N$, given a permission structure $\varrho$ and a weight $\omega\in \mathbb{R}^{|N|}_+$, the weighted permission Shapley value with mixed approach for a player $i\in N$ is
\[ \varphi^{\omega,\varrho}_i(v) = \varphi^{\omega}_i(v^\varrho). \]
% \[ \phi^{\omega,\varrho}_i(v) = \sum_{E\subseteq A_{\varrho}} \left( \sum_{\substack{F: E\in \Lambda^*(F)}} d^{\mathcal{P}_{\varrho}(u_F)}(E) \cdot d^v(F) \right)\phi^{\omega}_i(u_E) \]
\end{definition}

As an example, if we apply the weighted permission Shapley value with mixed approach on the diffusion cooperative game in Example~\ref{example:dag} and letting $\omega = 1^{|N|}$, then the reward distributed to each player is %: $\phi^{\omega,\varrho}_{1} = 11/15$, $\phi^{\omega,\varrho}_{2} = 2/3$, $\phi^{\omega,\varrho}_3 = \phi^{\omega,\varrho}_5 = \phi^{\omega,\varrho}_6 = \phi^{\omega,\varrho}_7 = 1/15$ and $\phi^{\omega,\varrho}_4 = 1/3$.
\[\begin{cases}
\varphi^{\omega,\varrho}_{1} = 11/15, \qquad
\varphi^{\omega,\varrho}_{2} = 2/3, \\
\varphi^{\omega,\varrho}_3 = \varphi^{\omega,\varrho}_5 = \varphi^{\omega,\varrho}_6 = \varphi^{\omega,\varrho}_7 = 1/15, \\
\varphi^{\omega,\varrho}_4 = 1/3. \\
\end{cases}\]
%Intuitively, the reward is distributed on all paths from the initial players to the other players.
Similar to the mechanisms in a forest, we can conclude that weighted Shapley value with mixed approach is a desirable mechanism that satisfies efficiency and diffusion incentive compatibility if the weight function is selected properly.

\begin{definition}
A weight function $\omega_i$ is proper if it only depends on $d_i$ as $\omega_i = f(d_i)$, where $d_i$ is the distance of player $i$ to the initial players $\mathcal{I}$ in the graph, i.e. the minimum distance between $i$ to one of the initial players ($\min_{j\in \mathcal{I}}\{d_{ji}\}$) and $f: \mathbb{N}\rightarrow \mathbb{R}_+$ is monotone non-decreasing.
\end{definition}

\begin{theorem}
For the monotone diffusion cooperative game in a DAG, if the reward distribution mechanism $\mathcal{M}$ is the weighted permission Shapley value with mixed approach with a proper weight function $\omega_i = f(d_i)$, %where $d_i$ is the distance of player $i$ to the initial players $\mathcal{I}$ in the graph, i.e. the minimum distance between $i$ to one of the initial players ($\min_{j\in \mathcal{I}}\{d_{ji}\}$) and $f: \mathbb{N}\rightarrow \mathbb{R}_+$ is monotone non-decreasing,
then $\mathcal{M}$ is efficient and DIC.
\end{theorem}

\begin{proof}
The efficiency can be easily derived since for all $\theta'\in \Theta$, $v^{\varrho}(J_{\mathcal{I}}(G(\theta'))) = v(J_{\mathcal{I}}(G(\theta')))$.

For the property of DIC, if we consider each player $i$ and edge $e=(i,j)\in E$, there are two cases that may happen if $i$ does not invite $j$ given any possible report profile of others $\theta_{-i}'\in \Theta_{-i}$.

(i) if $j$ cannot join the coalition, i.e., $j\notin J_{\mathcal{I}}(G(\theta'))$, then the proof is similar to that of Theorem~\ref{thm:tree} that shows player $i$ will not get more reward without inviting $j$.

(ii) if $j$ still can join the coalition, i.e., $j\in J_{\mathcal{I}}(G(\theta'))$, Let $v^{\varrho}$ be the projection game when $i$ invites $j$ and $v^{\varrho'}$ be the projection game when $i$ does not invite $j$. Suppose $R$ is an order in $\mathcal{R}(N)$. Then we have %$C_i(v^{\varrho}, R) \geq C_i(v^{\varrho'}, R)$
% \[ C_i(v',R) \geq C_i(v'',R) \]
\[
\begin{cases}
C_i(v^{\varrho}, R) = C_i(v^{\varrho'}, R) & \text{if } i \text{ comes before } j \text{ in } R; \\
C_i(v^{\varrho}, R) \geq C_i(v^{\varrho'}, R)& \text{if } j \text{ comes before } i \text{ in } R.
\end{cases}
\]
The above (in)equalities hold because (1) if $i$ is at the position before $j$, the marginal contribution of $i$ is unchanged; (2) if $i$ is at the position after $j$, she cannot bring $j$'s contribution when she does not invite $j$. Note that $d_k$ will not change for any player $k$ with $d_k< d_j$. Let $d_j'$ be the distance of player $j$ to initial players if $i$ does not invite $j$ and hence $d_j'\geq d_j$. Thus, (1) if $d_j' = d_j$, then the weights of all players will not change and so do the weights of the orders, which can be computed from weights $\omega$. Hence,
\[ \varphi_i^{\omega, \varrho} = \sum_{R\in \mathcal{R}(N)} \omega_R C_i(v^{\varrho},R) \geq \sum_{R\in \mathcal{R}(N)} \omega_R C_i(v^{\varrho'},R) = \varphi_i^{\omega, \varrho'} \]
where $\varphi_i^{\omega, \varrho}$ is player $i$'s reward when she invites $j$ and $\varphi_i^{\omega, \varrho'}$ is player $i$'s reward when she does not invite $j$. (2) if $d_j' > d_j$, then $f(d_j)\leq f(d_j')$ since $f$ is monotone non-decreasing. %note that we must have $d_i < d_j$.
%since the function $f$ is monotone non-increasing, then $f(d_i)\geq f(d_j)\geq f(d_j')$.
Let $R_{ij}\in \mathcal{R}(N)$ be some order where $i$ comes before $j$ and $R_{ji}$ is the corresponding order where $i$ and $j$'s positions are exchanged in $R_{ij}$. %According to the computation for weights of orders (as in Theorem~\ref{thm:tree}),
We have $\frac{\omega_{R_{ji}}}{\omega_{R_{ij}}}\geq \frac{\omega'_{R_{ji}}}{\omega'_{R_{ij}}}$ (since $R_{ij}$ is more likely sampled than $R_{ji}$ with a larger $\omega(j)$). Hence, we have
\begin{align*}
\varphi_i^{\omega, \varrho} & = \sum_{R_{ij}\in \mathcal{R}(N)} \left[ \omega_{R_{ij}} C_i(v^{\varrho}, R_{ij}) + \omega_{R_{ji}} C_i(v^{\varrho}, R_{ji}) \right] \\
& \geq \sum_{R_{ij}\in \mathcal{R}(N)} \left[ \omega'_{R_{ij}} C_i(v^{\varrho'},R_{ij}) + \omega'_{R_{ji}} C_i(v^{\varrho'},R_{ji}) \right] \\
& = \varphi_i^{\omega, \varrho'}.
\end{align*}
% \begin{align*}
% \varphi_i^{\omega, \varrho} & = \sum_{R_{ij}\in \mathcal{R}(N)} \left[ \omega_{R_{ij}}C_i(\mathcal{P}_{\varrho}(v),R_{ij}) + \sum_{R_{ji}\in \mathcal{R}(N)} \omega_{R_{ji}}C_i(\mathcal{P}_{\varrho}(v),R_{ji}) \right] \\
% & \geq \sum_{R_{ij}\in \mathcal{R}(N)} \omega'_{R_{ij}}C_i(\mathcal{P}_{\varrho'}(v),R_{ij}) + \sum_{R_{ji}\in \mathcal{R}(N)} \omega'_{R_{ji}}C_i(\mathcal{P}_{\varrho'}(v),R_{ji}) \\
% & = \bar{\phi}_i^{\omega, \varrho'}
% \end{align*}
The inequality in the second line holds since $ \omega_{R_{ij}} +  \omega_{R_{ji}} =  \omega'_{R_{ij}} +  \omega'_{R_{ji}}$, %$\sum_{R_{ij}} \omega_{R_{ij}} $ $+$ $ \sum_{R_{ji}} \omega_{R_{ji}}$ $=$ $\sum_{R_{ij}} \omega'_{R_{ij}}$ $+$ $\sum_{R_{ji}} \omega'_{R_{ji}}$ $=$ $1$, 
$C_i(v^{\varrho}, R_{ij})\geq C_i(v^{\varrho' }, R_{ji})$ and $C_i(v^{\varrho}, R_{ij}) = C_i(v^{\varrho'}, R_{ij})$ (i.e., the larger term will obtain a larger factor). Therefore, in all cases $i$ will not invite fewer agents. As a result, $\mathcal{M}$ is DIC.
\end{proof}

Finally, it is worth to point out that the weighted permission Shapley value that represents the solution of DARPA winning team also has a monotone non-decreasing weight function $\omega$. Hence, it can be seen as \textbf{a diffusion incentive compatible extension in DAGS of DARPA winning team's solution}.

\subsection{Applying on General Graphs}
We discuss the possibility to apply our method on general graphs. One may observe that in some real scenarios, the DAG we modelled is not necessarily the underlying social network. The underlying network could be any undirect graph. This paper focuses on the DAG because in practice, a DAG could be the result of players' invitations associated with timestamp. Another reason to use DAG is that it is intuitive and handy to define permissions, which clearly specifies who permits/invites who.

It is worthwhile to point out that actually, our results can be extended to general undirected graphs. The only difficulty is defining the permission structure but there are many possible ways. One way we can provide here is to define the permission set of an agent $i$ to be all her neighbors via whom $i$ can reach one of the sources with a simple path (this also works even if there are cycles), i.e., $P(i) = \{j\mid (i,j)\in E \text{ and there exists a simple path } i\rightarrow j\rightarrow\dots \rightarrow s \text{ such that } s\in \mathcal{I} \}$. Then, all results in this section can still hold.

\section{Conclusion}
In this paper, we formalize the problem of diffusion incentives in cooperative games for the first time. We design a family of reward distribution mechanisms such that players are incentivized to invite their neighbours to join the coalition. The family of reward distribution mechanisms combines the idea of the Shapley value with permission structure and weight system, which well explains the classic solution given by the winning team of DARPA 2009 red balloon challenge. We expect that our work will have very promising applications via social networks such as resource acquisition and question answering. One interesting future direction is to characterize the necessary and sufficient conditions for diffusion incentive compatibility.% We also believe that our solutions cover almost all DIC mechanisms for DAGs.

%%%%%%%%%%%%%%%%%%%%%%%%%%%%%%%%%%%%%%%%%%%%%%%%%%%%%%%%%%%%%%%%%%%%%%%%

%%% The acknowledgments section is defined using the "acks" environment
%%% (rather than an unnumbered section). The use of this environment 
%%% ensures the proper identification of the section in the article 
%%% metadata as well as the consistent spelling of the heading.

% \begin{acks}
% If you wish to include any acknowledgments in your paper (e.g., to 
% people or funding agencies), please do so using the `\texttt{acks}' 
% environment. Note that the text of your acknowledgments will be omitted
% if you compile your document with the `\texttt{anonymous}' option.
% \end{acks}

%%%%%%%%%%%%%%%%%%%%%%%%%%%%%%%%%%%%%%%%%%%%%%%%%%%%%%%%%%%%%%%%%%%%%%%%

%%% The next two lines define, first, the bibliography style to be 
%%% applied, and, second, the bibliography file to be used.

\bibliographystyle{ACM-Reference-Format} 
\bibliography{sample}

%%%%%%%%%%%%%%%%%%%%%%%%%%%%%%%%%%%%%%%%%%%%%%%%%%%%%%%%%%%%%%%%%%%%%%%%

\end{document}